\documentclass[11pt,a4paper,final]{amsart}

\usepackage{amsmath}
\usepackage{paralist}
\usepackage{graphics}
\usepackage{graphicx}
\usepackage{epstopdf}
\usepackage[colorlinks=true]{hyperref}
\usepackage{amssymb}
\usepackage{amsthm}
\usepackage{color}
\usepackage{multirow}
\usepackage{booktabs}
\usepackage{bm}
\usepackage{mathrsfs}

\newtheorem{theorem}{theorem}[section]
\newtheorem{corollary}{corollary}
\newtheorem{lemma}[theorem]{lemma}
\newtheorem{proposition}{proposition}

\newtheorem{definition}[theorem]{definition}

\newcommand{\ud}{\mathrm{d}}
\newcommand{\R}{\mathbb{R}}

\begin{document}

\title[Multidimensional linear and nonlinear PIDE]{Multidimensional linear and nonlinear partial integro-differential equation in Bessel potential spaces with applications in option pricing}

\author{Daniel \v{S}ev\v{c}ovi\v{c}${}^{1}$}
\author{Cyril Izuchukwu Udeani${}^{1}$}
\address{${}^{1}$ Department of Applied Mathematics and Statistics, Faculty of Mathematics Physics and Informatics, Comenius University, Mlynsk\'a dolina, 842 48, Bratislava, Slovakia. Corresponding author: {\tt sevcovic@fmph.uniba.sk, cyril.udeani@fmph.uniba.sk} }

\begin{abstract}
The purpose of this paper is to analyze solutions of a non-local nonlinear partial integro-differential equation (PIDE) in multidimensional spaces. Such class of PIDE often arises in financial modeling. We employ the theory of abstract semilinear parabolic equations in order to prove existence and uniqueness of  solutions in the scale of Bessel potential spaces. We consider a wide class of  L\'evy measures satisfying suitable growth conditions near the origin and infinity. The novelty of the paper is the generalization of already known results in the one space dimension to the multidimensional case. We consider Black-Scholes models for option pricing on underlying assets following a L\'evy stochastic process with jumps. As an application to option pricing in the one-dimensional space, we consider a general shift function arising from nonlinear option pricing models taking into account a large trader stock-trading strategy. We prove existence and uniqueness of a solution to the nonlinear PIDE in which the shift function may depend on a prescribed large investor stock-trading strategy function.

\medskip
\noindent
2010 MSC. Primary: 45K05 · 35K58 · 34G20 · 91G20

\noindent Key words and phrases. L\'evy measure; Option pricing; Strong kernel; H\"older continuity; Partial integro-differential equation; Bessel potential spaces; 

\end{abstract}

\maketitle

\section{Introduction}

Our goal is to prove existence and uniqueness of a solution to a nonlinear parabolic partial integro-differential equation (PIDE) having the form: 
\begin{eqnarray}
\frac{\partial u}{\partial \tau} &=& \Delta u 
+ \int_{\mathbb{R}^n}\left[ u(\tau, x+z)-u(\tau, x)- z\cdot \nabla_x u(\tau,x)
\right] \nu(\ud z) +  g(\tau, x, u, \nabla_x u), 
\label{PDE-u}
\\
&&u(0,x)=u_0(x), \quad x\in \R^n, \tau\in(0,T),
\nonumber
\end{eqnarray}
where $g$ is a given sufficiently smooth function. By $\nu $, we denote a positive measure on $\mathbb{R}^n$ such that its Radon derivative is a non-negative Lebesgue measurable function $h$ in $\R^n$, i.e., $\nu (\ud z) = h(z) \ud z$.

Recently, the theory of stochastic processes have attracted significant attention owing to their important mathematical development and applications. Stochastic processes are mathematically rich objects with wide range of applications, e.g., in engineering, physics, and economics. It is worth noting that these processes occur in almost every discipline. Various stochastic processes, such as random walks, Markov chains, martingales, and L\'evy processes have attracted significant attention due to their wide range of applications. In particular, the concept of L\'evy measures plays a significant role owing to its theoretical developments and wealth of novel applications, e.g., to option pricing in financial mathematics. 
In the probability theory, the Laplace operator occurs as an infinitesimal generator of a stochastic process since it generates a semigroup corresponding to a symmetric stable process. The relationships between more general non-local operators and jump processes has been widely investigated. In particular, there is a relationship between the solution to PIDE and properties of the corresponding Markov jump process (c.f.,  Abels and Kassmann \cite{AbelsKass2009}, or Florescu and Mariani \cite{florescu2010solutions}).

In the past decades, the role of PIDEs has been investigated in various fields, such as pure mathematical areas, biological sciences, and economics \cite{aboodh2016solution,florescu2010solutions,NBS15,yuzbacsi2016improved}. PIDE problems arising from financial mathematics, especially from option pricing models, have been of great interest to many researchers. In most cases, standard methods for solving these problems lead to the study of parabolic equations. Mikulevi\v{c}ius and Pragaraustas \cite{mikulevivcius1992cauchy} investigated solutions of the Cauchy problem to the parabolic PIDE with variable coefficients in Sobolev spaces. The results were applied in order to obtain solutions of the corresponding martingale problem. Crandal \emph{et al.} \cite{ishii1996viscosity} employed the notion of a viscosity solution to investigate the qualitative results. Their results were extended and generalized by Soner \emph{et al.} \cite{burzoni2020viscosity} and Barles \emph{et al.} \cite{barles1997backward} for the first and second order operators, respectively. Florescu and Mariani \cite{florescu2010solutions} utilized the Schaefer fixed point theorem to establish existence of a weak solution of the generalized PIDE. Amster \emph{et al.} \cite{Amster12} used the method of upper and lower solutions. They proved the existence of solutions in a general domain for multiple assets and the  regime switching jump diffusion model. Cont \emph{et al.} \cite{cont2005integro} investigated the actual connection between option pricing in exponential L\'evy models and the corresponding PIDEs for European options and those with single or double barriers. They discussed the conditions for which prices of option are classical solution of the PIDE. Cruz and \v{S}ev\v{c}ovi\v{c} \cite{CruzSevcovic2020} employed the theory of abstract semilinear parabolic equations to obtain a solution to the nonlinear non-local PIDE (\ref{PDE-u}) in the one-dimensional case. They employed the framework of a scale of Bessel potential spaces for general L\'evy measures by considering classical Black-Scholes model that depends on various restricted assumptions, e.g., liquidity and completeness of the market. They relaxed the complete market hypothesis and assumed a L\'evy process for underlying stock price to obtain a model for pricing European and American call and put options on an underlying asset characterized by a L\'evy measure.

In this paper, we extend and generalize the results of Cruz and \v{S}ev\v{c}ovi\v{c} to a multidimensional space in the scale of Bessel potential Sobolev spaces. We shall also analyze the following generalization of PIDE:
\begin{equation}
\frac{\partial u}{\partial \tau} = \frac{\sigma^2}{2} \Delta u 
+ \int_{\mathbb{R}^n}\left[ u(\tau, x+ \xi )-u(\tau, x)- \xi \cdot \nabla_x u(\tau,x)
\right] \nu(\ud z) +  g(\tau, x, u, \nabla_x u) ,
\label{PDE-u-general}
\end{equation}
where $\xi=\xi(\tau, x, z)$ is a shift function, which may depend on the variables $\tau>0, x,z\in\R$. An application for such a general shift function $\xi$ can be found in nonlinear option pricing models considering a large trader stock-trading strategy with the underlying asset price dynamic following the L\'evy process (c.f.,  Cruz and \v{S}ev\v{c}ovi\v{c} \cite{NBS19}). If $\xi(x, z)\equiv z$, then (\ref{PDE-u-general}) reduces to equation (\ref{PDE-u}). The nonlinearity $g$ often arises from applications occurring  in pricing e.g., XVA derivatives (c.f.,  Arregui \emph{et al.} \cite{NBS15, NBS17}) or applications of the penalty method for American option pricing under a PIDE model (c.f.,   Cruz and \v{S}ev\v{c}ovi\v{c} \cite{CruzSevcovic2020}).

Our aim is to investigate solutions to PIDE (\ref{PDE-u}) in the framework of Bessel potential spaces for a multidimensional case, $n\ge 1$. They form a nested scale $\{X^\gamma\}_{\gamma\ge0}$ of Banach spaces satisfying $X^{\gamma_1} \hookrightarrow X^{\gamma_2}$ for any $1\geq \gamma_1\geq \gamma_2\geq 0$, and $X^1\equiv D(A), X^0\equiv X$. Here, the operator $A$ is sectorial in the space $X$ having a dense domain $D(A)\subset X$ (c.f.,  Henry \cite{Henry1981}). An example of such an operator is the Laplace operator, i.e., $A=-\Delta$ in $\R^n$ with the domain $D(A)\equiv W^{2,p}(\R^n)\subset X\equiv L^p(\R^n)$. It is known that if $A=-\Delta$, then $X^\gamma$ is embedded in the Sobolev-Slobodecki space $W^{2\gamma,p}(\R^n)$, which is a space consisting of all functions  such that $2\gamma$-fractional derivative belongs to the Lebesgue space $L^p(\R^n)$ of $p$-integrable functions (c.f.,  \cite{Henry1981}). We investigate solutions to the PIDE (\ref{PDE-u-general})  for a wide class of L\'evy measures $\nu$  satisfying suitable growth conditions near $\pm\infty$ and origin in a higher dimensional space. 

The remainder of the paper is organized as follows. In Section 2, we present preliminaries and basic notions of L\'evy measures, and introduce concept of admissible activity L\'evy measures  with some important examples. In Section 3, we establish the existence results to the nonlinear non-local PIDE (\ref{PDE-u-general}) in the framework of Bessel potential spaces for the multidimensional case. We analyze the generalized nonlinear non-local PIDE (\ref{PDE-u-general}), where the shift function $\xi=\xi(\tau, x, z)$ depends on the variables $x,z\in\R^n$. Section 4 presents an application of the proposed results in the one-dimensional space for pricing of options on underlying asset that follows L\'evy processes. Section 5 presents the conclusions. 

\section{Preliminaries, definitions and motivation }
In this section, we presents some basic definitions of L\'evy measures, and notion of admissible activity L\'evy measures. In what follows, we denote the Euclidean norm in $\R^n$ and the norm in an infinite dimensional function space (e.g., $L^p(\R^n), X^\gamma$) by $|\cdot|$, and $\Vert\cdot\Vert$, respectively. Henceforth, $a\cdot b$ stands for the usual Euclidean product in $\R^n$ with the norm $|z| = \sqrt{z\cdot z}$.

A L\'evy process on $\R^n$ is a stochastic (right continuous) process $X = \{X_t, t\geq 0\}$ having the left limit with independent stationary increments. It is uniquely characterized by its L\'evy exponent $ \phi $:
\[
\mathbb{E}_x(e^{i y\cdot X_t}) = e^{-t\phi(y)}, ~ y\in\R^n.
\]
The subscript $x$ in the expectation operator $\mathbb{E}_x$ indicates that the process $X_t$ starts from a given value $x$ at the origin $t=0$. The L\'evy exponent $\phi$ has the following unique decomposition: 
\[
\phi(y) = ib\cdot y + \sum_{i, j=1}^n a_{ij} y_i y_j + \int_{\R^n}\left(1-e^{iy\cdot z} + iy\cdot  z1_{|z|\leq 1}\right)\nu(\ud z), 
\]
where $b\in\R^n$ is a constant vector, $(a_{ij})$ is a constant matrix, which is positive semidefinite; $\nu(\ud z)$ is a nonnegative measure on $\R^n \setminus \{0\}$ such that $\int_{\R^n}\min(1, |z|^2)\nu(\ud z) <\infty$ \\
(c.f.,  \cite{palatucci2017recent}).
\subsection{Admissible activity L\'evy measures}

Next, we recall notion of an admissible activity L\'evy measure introduced by Cruz and \v{S}ev\v{c}ovi\v{c} \cite{NBS19, CruzSevcovic2020} for the one-dimensional case $n=1$. We shall extend this notion  for the multidimensional case $n\ge 1$.

\begin{definition}\label{def-admissiblemeasure}
A measure $\nu$ in $\R^n$ is called an admissible activity L\'evy measure if  there exists a non-negative Lebesgue measurable function $h:\R^n\to \R$ such that $\nu(\ud z) = h(z) \ud z$ such that 
\begin{equation}
0 \le  h(z)\le C_0 | z|^{-\alpha} e^{- D |z| - \mu |z|^{2}},
\label{growth_measure}
\end{equation}
for all $z\in\R^n$ and the shape parameters $\alpha, \mu\geq 0, D\in \R$, ($D>0$ if $\mu=0)$. Here, $C_0>0$ is a positive constant. 
\end{definition}

It is worth noting that the class of admissible activity L\'{e}vy measures consists of various important measures often used in financial modeling of underlying stock dynamics with jumps. For instance, in the context of financial modeling, one of first jump-diffusion models was proposed by Merton \cite{Merton76}. Its L\'{e}vy measure has the form:
\begin{equation}
\nu(\ud z)=\lambda \frac{1}{(2\pi\delta^2)^{n/2}}e^{-\frac{|z-m|^2}{2\delta^2}}\ud z\,.
\label{merton-density}
\end{equation}
Here, the parameters $m\in\R^n, \lambda, \delta>0,$ are given. Merton's measure is a finite activity L\'evy measure, i.e., $\nu(\mathbb{R}^n)=\int_{\mathbb{R}^n}\nu(\ud z)<\infty$, having finite variation $\int_{|z|\leq 1}|z|\nu(\ud z)<\infty$. 

A class of admissible L\'evy measures also contains measures of the form:
\begin{equation}
\nu (\ud z)= C_0 |z|^{-\alpha}e^{-A|z|}\ud z, 
\label{vargamma-density}
\end{equation} 
where $\alpha\in\R, A> 0$. Such a L\'evy  measure is of finite activity provided that $\alpha < n$, it has a finite variation if  $\alpha < n+1$. 

In the one-dimensional case $n=1$, the class of admissible L\'evy measures includes the double exponential model introduced by Kou \cite{Kou2002}, Variance Gamma model \cite{DPE98}, Normal Inverse Gaussian (NIG), and CGMY models \cite{BARNIE01}. We refer the reader to the paper \cite{CruzSevcovic2020} by Cruz and \v{S}ev\v{c}ovi\v{c}  for an  overview of various examples of admissible L\'evy measures in the one-dimensional case.

\section{Existence and uniqueness results}

In this section, we focus on the proof of the existence and uniqueness of a solution $u(\tau,x)$ to (\ref{PDE-u-general}) for a wide class of L\'evy measures. The PIDE (\ref{PDE-u-general}) can be rewritten as follows:
\begin{eqnarray}
&&\frac{\partial u}{\partial \tau} + A u = 
 f(u) + g(\tau, x, u, \nabla_x u), \quad u(0,x)=u_0(x), \quad x\in \mathbb{R}^n, \tau\in(0,T),
\label{problem_transformed}
\end{eqnarray}
where  $A = -(\sigma^2/2) \Delta $. The linear non-local operator $f$ is defined by:
\begin{equation}
f(u)(\cdot) =
\int_{\mathbb{R}^n}\left[ u(\cdot+\xi)-u(\cdot)- \xi \cdot \nabla_x u(\cdot)\, \right] \nu(\ud z),
\label{functional_f_def}
\end{equation}
where $\xi = \xi(\tau,x,z)$ is a given shift function. The mapping $g$ is assumed to be H\"older continuous in the $\tau$ variable and Lipschitz continuous in the remaining  variables.  The proof of the existence, continuation, and uniqueness of a solution is based on the abstract theory of semilinear parabolic equations according to Henry \cite{Henry1981}. A solution to the PIDE (\ref{problem_transformed}) is constructed in the scale of the Bessel potential spaces ${\mathscr L}^p_{2\gamma}(\mathbb{R}^n),\gamma\ge 0$, for the multidimensional case, $n\geq 1$. These spaces can be considered as a natural extension of the classical Sobolev spaces $W^{k,p}(\mathbb{R}^n)$ for non-integer values of order $k$. The nested scale of Bessel potential spaces allows for a finer formulation of existence and uniqueness results compared to the classical Sobolev spaces.

In what follows, we briefly recall the construction and basic properties of Bessel potential spaces. Recall that if $A$ is a sectorial operator in a Banach space $X$, then $-A$ is a generator of an analytic semigroup $\left\{e^{-A t}, t\geq 0\right\}$ acting on $X$ (c.f.,  \cite[Chapter I]{Henry1981}). For any $\gamma >0$, we can introduce the operator $A^{-\gamma}:X\to X$ as follows: $A^{-\gamma}=\frac{1}{\Gamma(\gamma)}\int_{0}^{\infty} \xi^{\gamma-1}e^{-A \xi} \ud \xi$. Then, the fractional power space $X^{\gamma} = D(A^{\gamma})$ is the domain of the operator $A^\gamma= (A^{-\gamma})^{-1}$. That is, $X^{\gamma}=\left\{u\in X:\  \exists \varphi\in X,  u=A^{-\gamma}\varphi\right\}$. The norm is defined as follows: $\Vert u\Vert_{X^\gamma}=\Vert A^{\gamma}u\Vert_X=\Vert \varphi\Vert_X$. Furthermore, we have continuous embedding: $D(A)\equiv X^1 \hookrightarrow X^{\gamma_1} \hookrightarrow X^{\gamma_2} \hookrightarrow X^0\equiv X$, for $0\le \gamma_2\le \gamma_1\le 1$. 

Let us recall the convolution operator $(G*\varphi)(x)=\int_{\mathbb{R}^n} G(x-y)\varphi(y)\ud{y}$.  According to \cite[Section 1.6]{Henry1981}, $ A= -(\sigma^2/2) \Delta$ is a sectorial operator in the Lebesgue space $X=L^{p}(\mathbb{R}^{n})$ for any $p\ge 1, n\ge 1$, and $D(A) \subset W^{2,p}(\mathbb{R}^{n})$. Now, it follows from  \cite[Chapter 5]{Stein1970} that the space $X^\gamma, \gamma>0,$ can be identified with the Bessel potential space ${\mathscr L}^p_{2\gamma}(\mathbb{R}^n)$, where 
\[
{\mathscr L}^p_{2\gamma}(\mathbb{R}^n):=\{u\in X:\  \exists \varphi\in X, u=G_{2\gamma}*\varphi\}.
\]
Here, $G_{2\gamma}$ is the Bessel potential function, 
\[
G_{2\gamma}(x) = \frac{1}{(4\pi)^{n/2}\Gamma(\gamma)} \int_0^\infty  y^{-1+\gamma-n/2} e^{-(y +|x|^2/(4y))}\ud y .
\]
The norm of $u=G_{2\gamma}*\varphi$ is given by $\Vert u\Vert_{X^\gamma}=\Vert \varphi\Vert_{L^p}$. The space $X^\gamma$ is continuously embedded in the fractional  Sobolev-Slobodeckii space $W^{2\gamma,p}(\R^n)$ (c.f.,  \cite[Section 1.6]{Henry1981}).

By $C_0>0$, we shall denote a generic constant which is independent of the solution $u$; however, it may depend on model parameters, e.g., $n\ge 1, p\ge 1, \gamma\in[0,1)$.

\begin{proposition}\label{prop_pointwise_est-2}
Let us define the mapping $Q(u,\xi)$ as follows:
\[
Q(u,\xi) = u(x +\xi(x))-\xi(x)\cdot\nabla_x u(x), \quad x\in\R^n.
\]
Then, there exists a constant $\hat{C}>0$ such that, for any vector valued functions  $\xi_1, \xi_2\in (L^\infty(\R^n))^n$, and $u$ such that $\nabla_x u \in (X^{\gamma-1/2})^n$, $1/2\le \gamma<1$, the following estimate holds:
\[
\Vert Q(u,\xi_1) - Q(u,\xi_2)  \Vert_{L^p({\R^n})}
\leq \hat{C} \Vert\xi_1-\xi_2\Vert_\infty^{2\gamma-1} (\Vert\xi_1\Vert_\infty+\Vert\xi_2\Vert_\infty) \Vert\nabla_x u\Vert_{X^{\gamma-1/2}}.
\]
\end{proposition}

\begin{proof}
Let $u\in X$ be such that  $\nabla_x u\in (X^{\gamma-1/2})^n$, i.e., $\partial_{x_i}u\in X^{\gamma-1/2} ~ \text{for each} ~ i = 1, \cdots, n$. Then $\nabla_x u= A^{-(2\gamma-1)/2}\varphi  = G_{2\gamma-1} * \varphi$ for some $\varphi\in (L^p(\R^n))^n$, and $\Vert \nabla_x u\Vert_{X^{\gamma-1/2}} = \Vert A^{(2\gamma-1)/2}\nabla_x u\Vert_X = \Vert \varphi\Vert_{L^p}$. 
Here, $\varphi = (\varphi_1, \cdots, \varphi_n)$ and $ \partial_{x_i}u =  G_{2\gamma-1} * \varphi_i$. Let $x, \xi\in \R^n$. Then 
\[
\nabla_x u(x+\xi) = G_{2\gamma-1}(x+\xi - \cdot)* \varphi(\cdot), \qquad  \nabla_x u(x)
= G_{2\gamma-1}(x - \cdot)* \varphi(\cdot).
\]
Recall that the convolution operator satisfies the following inequality:
\[
\Vert \psi*\varphi\Vert_{L^p(\R^n)}\le \Vert \psi\Vert_{L^q(\R^n)} \Vert \varphi\Vert_{L^r(\R^n)},
\]
where  $p,q,r\ge 1$ and $1/p + 1 = 1/q + 1/r$ (see \cite[Section 1.6]{Henry1981}). In particular, for $q=1$ we have $\Vert \psi*\varphi\Vert_{L^p}\le \Vert \psi\Vert_{L^1} \Vert \varphi\Vert_{L^p}$. 
For the modulus of continuity of the Bessel potential function $G_{\alpha}, \alpha\in(0,1)$, the following estimate holds:
\[
\Vert G_{\alpha}(\cdot + h) - G_{\alpha}(\cdot) \Vert_{L^1}
\le C_0 |h|^{\alpha},
\]
for any $h\in \R^n$ (c.f.,  \cite[Chapter 5.4, Proposition 7]{Stein1970}).
Let $\xi_1, \xi_2$ be bounded vector valued  functions, i.e., $\xi_1, \xi_2\in (L^\infty(\R^n))^n$. Then, for any $x\in \R^n ~ \text{and} ~\theta\in [0,1]$, we have
\begin{eqnarray*}
&&
u(x+\xi_1(x)) - u(x+\xi_2(x)) - (\xi_1(x) -\xi_2(x)) \cdot\nabla_x u(x)
\\
&=& u(x+\xi_1(x)) - \nabla_x u(x) + \xi_1(x)\cdot \nabla_x u(x) 
\\
&& - [u(x+\xi_2(x))- \nabla_x u(x) -\xi_2(x)\cdot \nabla_x u(x)]
\\
&=& (\xi_1(x) -\xi_2(x))\int_0^1\nabla_x u(x+\xi_1(x)) - \nabla_x u(x)\ud{\theta}  
\\
&&
+ \int_0^1 \nabla_x u(x+\theta \xi_1(x))-  \nabla_x u(x+\theta \xi_2(x))\ud{\theta}
\ .
\end{eqnarray*}
Now,
\begin{eqnarray*}
&&
\Vert Q(u,\xi_1) - Q(u,\xi_2)  \Vert^p_{L^p({\R^n})}
\\
&=& \int_{\R^n} |u(x+\xi_1(x)) - u(x+\xi_2(x)) - (\xi_1(x) - \xi_2(x)) \cdot\nabla_x u(x)|^p\ud{x}
\\
&\le&  \int_{\R^n} \left|(\xi_1(x)-\xi_2(x)) \int_0^1 \nabla_x u(x+\theta \xi_1(x) )-  \nabla_x u(x)\ud{\theta} \right|^p \ud{x}
\\
&& + \int_{\R^n} \left|\xi_2(x) \int_0^1 \nabla_x u(x+\theta \xi_1(x))- \nabla_x u(x+\theta\xi_2(x))\ud{\theta} \right|^p \ud{x}
\\
&\le& \Vert\xi_1 - \xi_2\Vert_{\infty}^p\int_0^1 \int_{\R^n} |\nabla_x u(x+\theta\xi_1(x))  - \nabla_x u(x)|^p dx d\theta 
\\
&& + \Vert\xi_2\Vert_{\infty}^p\int_0^1 \int_{\R^n} |\nabla_x u(x+\theta\xi_1(x))  - \nabla_x u(x+\theta\xi_2(x)|^p dx d\theta 
\\
&\le&  \Vert\xi_1 - \xi_2\Vert_{\infty}^p\int_0^1 \Vert\left( G_{2\gamma-1}(\cdot + \theta\xi_1) - G_{2\gamma-1}(\cdot)\right)*\varphi\Vert_{L^p}^p d\theta 
\\
&& + \Vert\xi_2\Vert_{\infty}^p\int_0^1 \Vert\left( G_{2\gamma-1}(\cdot + \theta\xi_1) - G_{2\gamma-1}(\cdot+\theta\xi_2)\right)*\varphi\Vert_{L^p}^p d\theta 
\\
&\le&  \Vert\xi_1 - \xi_2\Vert_{\infty}^p\int_0^1  \Vert G_{2\gamma-1}(\cdot + \theta\xi_1 ) - G_{2\gamma-1}(\cdot)\Vert^p_{L^1} d\theta \Vert\varphi\Vert_{L^p}^p 
\\
&& +  \Vert \xi_2\Vert_{\infty}^p \int_0^1 \Vert G_{2\gamma-1}(\cdot + \theta\xi_1) - G_{2\gamma-1}(\cdot+\theta\xi_2)\Vert^p_{L^1} d\theta \Vert\varphi\Vert_{L^p}^p
\\
&\le&  \left(\Vert\xi_1 - \xi_2\Vert_{\infty}^p \Vert\xi_1 \Vert^{(2\gamma-1)p}_{\infty}
+
\Vert \xi_2\Vert_{\infty}^p \Vert\xi_1-\xi_2 \Vert^{(2\gamma-1)p}_{\infty}\right)C_0^p\Vert\nabla_x u\Vert^p_{X^{\gamma-1/2}} 
\\
&\le&  \Vert\xi_1-\xi_2 \Vert^{(2\gamma-1)p}_{\infty}\left(\Vert\xi_1\Vert_{\infty}^p +\Vert\xi_2\Vert^{(2-2\gamma)p}_{\infty} \Vert\xi_1 \Vert^{(2\gamma-1)p}_{\infty}
+
\Vert \xi_2\Vert_{\infty}^p \right)C_0^p\Vert\nabla_x u\Vert^p_{X^{\gamma-1/2}} 
\ .
\end{eqnarray*}
By Young's inequality, we have $ab\leq \frac{a^\alpha}{\alpha}+ \frac{b^\beta}{\beta}$ for any $a, b\geq 0$, and $\alpha, \beta>1 $ with $1/\alpha + 1/\beta=1$  (c.f., \cite{brezis2010functional}). Set $ \alpha = 1/(2- 2\gamma), \beta =1/(2\gamma -1)$. Then, $1/\alpha+1/\beta=1$ and  we obtain 
$\Vert\xi_2\Vert^{(2-2\gamma)p}_{\infty} \Vert\xi_1 \Vert^{(2\gamma-1)p}_{\infty}\leq (2-2\gamma)\Vert \xi_2\Vert_{\infty}^p + (2\gamma -1)\Vert \xi_1\Vert_{\infty}^p\leq 2\Vert\xi_2\Vert_{\infty}^p + \Vert\xi_1\Vert_{\infty}^p$. Therefore, 
\begin{eqnarray*}
\Vert Q(u,\xi_1) - Q(u,\xi_2)  \Vert_{L^p({\R^n})}^p
&\le&  2\Vert\xi_1-\xi_2 \Vert^{(2\gamma-1)p}_{\infty}\left(\Vert\xi_1\Vert_{\infty}^p +
\Vert \xi_2\Vert_{\infty}^p \right)C_0^p\Vert\nabla_x u\Vert^p_{X^{\gamma-1/2}} 
\\
&\leq& 2 C_0^p \Vert\xi_1-\xi_2\Vert_\infty^{(2\gamma-1)p} (\Vert\xi_1\Vert_\infty+\Vert\xi_2\Vert_\infty)^p \Vert\nabla_x u\Vert^p_{X^{\gamma-1/2}}.
\end{eqnarray*}
Hence, the pointwise estimate holds with the constant $\hat{C} = 2^{1/p}C_0 >0$. 
\end{proof}

Applying  Proposition \ref{prop_pointwise_est-2} with $\xi_1=\xi ~ \text{and} ~ \xi_2=0$, we obtain the following corollary. 

\begin{corollary}\label{prop_pointwise_est}
Let $u$ be such that $\nabla_x u \in (X^{\gamma-1/2})^n$ where $1>\gamma\ge 1/2$. Then, for any  $\xi\in\R^n$, the following pointwise estimate holds:
\[
\Vert Q(u,\xi) \Vert_{L^p({\R^n})}\leq C_0|\xi|^{2\gamma} \Vert\nabla_x u\Vert_{X^{\gamma-1/2}}.
\]
\end{corollary}

Next, we consider the case when the non-local integral term depends on $x$ and $z$ variables. It is a generalization of the result \cite[Lemma 3.4]{NBS19} due to Cruz and \v{S}ev\v{c}ovi\v{c} proven for the case where $\xi(x,z)\equiv z$.

\begin{proposition}\label{prop-f} 
Suppose that the shift mapping $\xi=\xi(x,z)$ satisfies $\sup_{x\in\R} |\xi(x,z)|  \le C_0 |z|^\omega (1+ e^{D_0 |z|})$ for some constants $C_0>0, D_0\ge 0, \omega>0$ and any $z\in\R$. Assume $\nu$ is a L\'evy measure with the shape parameters $\alpha, D,$ and either $\mu>0, D\in\R$, or $\mu=0$ and $D>D_0\ge 0$.  Assume $1/2\le \gamma <1$, and $\gamma> (\alpha-n)/(2\omega)$. Then there exists a constant $C_0>0$ such that
\[
\Vert f(u)\Vert_{L^p} \le C_0 \Vert \nabla_x u\Vert_{X^{\gamma-1/2}},
\]
provided that $\nabla_x u \in (X^{\gamma-1/2})^n$. If $u\in X^\gamma$ then $\Vert f(u)\Vert_{L^p} \le C \Vert u\Vert_{X^\gamma}$, i.e., $f:X^\gamma\to X$ is a bounded linear operator.
\end{proposition}

\begin{proof}
The L\'evy measure $\nu(\ud z)$ is given by   $\nu(\ud z)= h(z) \ud z$. Let us denote the auxiliary function $\tilde h(z)=|z|^\alpha h(z)$. Then, $0\le \tilde h(z)\le C_0 e^{- D |z| - \mu |z|^{2}}$. Since $h(z) =  |z|^{-\alpha} \tilde h(z) = h_1(z) h_2(z)$, where $h_1(z)=|z|^{-\beta} \tilde h(z)^\frac12$ and $h_2(z)=|z|^{\beta-\alpha} \tilde h(z)^\frac12$.
Applying  Proposition~\ref{prop_pointwise_est} with $\xi_1=\xi, \xi_2=0$, and using the H\"older inequality, we obtain
\begin{eqnarray*}
\Vert f(u)\Vert_{L^p}^p 
&=& \int_{\R^n} \left|\int_{\R^n} ( u(x+\xi(x,z)) - u(x) - \xi(x,z)\cdot \nabla_x u(x)) h(z)\ud{z}  \right|^p\ud{x}
\\
&\le& \int_{\R^n} \int_{\R^n} \left| u(x+\xi(x,z)) - u(x) - \xi(x,z) \cdot \nabla_x u(x)\right|^p h_1(z)^p \ud{z} 
\\
&& \quad\times \left(\int_{\R^n} h_2(z)^q \ud{z}\right)^{p/q} \ud{x}
\\
&=& \int_{\R^n} \left(\int_{\R^n} \left| u(x+\xi(x,z)) - u(x) - \xi(x,z) \cdot\nabla_x u(x)\right|^p \ud{x}\right) h_1(z)^p \ud{z} 
\\
&& \quad\times \left(\int_{\R^n} h_2(z)^q \ud{z}\right)^{p/q} 
\\
&\le & 
C_0^p  \Vert \nabla_x u\Vert_{X^{\gamma-1/2}}^p
\int_{\R^n} |\xi(x,z)|^{2\gamma p} |z|^{-\beta p} \tilde h(z)^{p/2} \ud{z} \left(\int_{\R^n} h_2(z)^q \ud{z}\right)^{p/q}
\\
&\le & 
C_0^p  \Vert \nabla_x u\Vert_{X^{\gamma-1/2}}^p
\int_{\R^n} |z|^{(2\gamma\omega -\beta) p} \tilde h(z)^{p/2} \ud{z} \left(\int_{\R^n} h_2(z)^q \ud{z}\right)^{p/q}.
\end{eqnarray*}
Assuming $p,q\ge 1, 1/p +1/q=1$ are such that 
\[
(2\gamma\omega -\beta) p > -n, \qquad (\beta-\alpha) q = (\beta-\alpha) \frac{p}{p-1} >-n,
\]
then, the integrals 
$\int_{\R^n} |z| ^{(2\gamma\omega -\beta) p} \tilde h(z)^{p/2} \ud{z}$ and 
$\int_{\R^n} h_2(z)^q \ud{z} = \int_{\R^n} |z| ^{(\beta-\alpha) q} \tilde h(z)^{q/2} \ud{z}$
are finite, provided that the shape parameters satisfy: either $\mu>0, D\in\mathbb{R}$, or $\mu=0, D>D_0\ge 0$. As $\gamma>(\alpha-n)/(2\omega)$, there exists $\beta>1$
satisfying 
\[
\alpha-n+n/p < \beta < 2\gamma\omega +n/p.
\]
Therefore, there exists $C_0>0$ such that $\Vert f(u)\Vert_{L^p} \le C_0 \Vert \nabla_x u\Vert_{X^{\gamma-1/2}}$.
\end{proof}

Let $C([0,T],X^{\gamma})$ denote the Banach space consisting of continuous functions from $[0,T]$ to $X^\gamma$ with the maximum norm. The following proposition is due to Henry \cite{Henry1981} (see also Cruz and \v{S}ev\v{c}ovi\v{c} \cite{CruzSevcovic2020}).

\begin{proposition}\cite{Henry1981}, \cite[Proposition 3.5]{CruzSevcovic2020}
\label{semilinear_general_existence_result}
Suppose that the linear operator $-A$ is a generator of an analytic semigroup $\left\{e^{-At},t\geq 0\right\}$ in a Banach space $X$. Assume the initial condition $U_0$ belongs to the space $X^{\gamma}$ where $0\leq \gamma <1$.  Suppose that the mappings $F:[0,T]\times X^{\gamma}\to X$ and $h:(0,T]\to X$ are H\"older continuous in the $\tau$ variable, $\int_0^T \Vert h(\tau)\Vert_X \ud \tau <\infty$, and $F$ is Lipschitz continuous in the $U$ variable. Then, for any $T>0$, there exists a unique solution to the abstract semilinear evolution equation: $\partial_\tau  U+A U=F(\tau, U) +h(\tau)$ 
such that $U\in C([0,T],X^{\gamma}), U(0)=U_{0}, \partial_\tau U(\tau) \in X, U(\tau)\in D(A)$ for any $\tau\in (0,T)$. The function $U$ is a solution in the mild (integral) sense, i.e., 
$U(\tau) = e^{-A \tau} U_0 + \int_0^\tau e^{-A (\tau-s)} (F(s, U(s)) + h(s) ) \ud{s}$, $\tau\in[0,T]$.
\end{proposition}

Applying  Propositions~\ref{prop-f} and \ref{semilinear_general_existence_result}, we can state  the following result which is a nontrivial generalization of the result shown by \v{S}ev\v{c}ovi\v{c} and Cruz \cite{CruzSevcovic2020} in the one space dimension $n=1$. 

\begin{theorem}
\label{semilinear_existence_result}
Suppose that the shift mapping $\xi=\xi(x,z)$ satisfies $\sup_{x\in\R} |\xi(x,z)|  \le C_0 |z|^\omega (1+ e^{D_0 |z|}), z\in\R^n$, for some constants $C_0>0, D_0\ge 0, \omega>0$. Assume $\nu$ is an admissible activity L\'evy measure with the shape parameters $\alpha, D$, and, either $\mu>0, D\in\R$, or $\mu=0, D>D_0\ge 0$. Assume $1/2\le \gamma < 1$ and $\gamma>(\alpha-n)/(2\omega)$, $n\geq 1$. Suppose that $g(\tau, x, u,\nabla_x u)$ is H\"older continuous in the $\tau$ variable and  Lipschitz continuous in the remaining variables, respectively. Assume $u_0\in X^\gamma$, and $T>0$. Then, there exists a unique mild solution $u$ to PIDE (\ref{PDE-u-general}) satisfying $u\in C([0,T],X^{\gamma})$.
\end{theorem}

\section{Applications to option pricing}

The classical linear Black--Scholes model and its multidimensional generalizations have been widely used in the analysis of financial markets. According to this model, the price  $V=V(t,S)$ of an option on an underlying asset price $S$ at time $t\in[0,T]$ can be obtained from a solution to the linear Black--Scholes parabolic equation:
\begin{equation}
\label{eq6}
\frac{\partial V}{\partial t} + \frac{1}{2}\sigma^2 S^2\frac{\partial^2 V}{\partial S^2} + r S\frac{\partial V}{\partial S} - r V = 0, \quad V(T,S)=\Phi(S), 
\end{equation}
$t \in [0,T),S > 0$. The parameter $r > 0$ represents the risk-free interest rate of a bond, $\sigma>0$ is the historical volatility of the underlying asset price $S$ time series. The underlying asset price is assumed to follow the geometric Brownian motion $dS/S = \mu dt + \sigma dW$. Here, $\{W_t, t\ge 0\}$ represents the standard Wiener process. The terminal condition $\Phi(S)$ represents the pay-off diagram  at maturity $t = T$, $\Phi(S)= (S-K)^+$ (call option case) or  $\Phi(S)= (K-S)^+$ (put option case). 

In the multidimensional case, where the option price $V(t, S_1, \cdots, S_n)$ depends on the vector of $n$ underlying stochastic assets $S=(S_1, \cdots, S_n)$  with the volatilities $\sigma_i$ and mutual correlations $\varrho_{ij}, i,j=1,\cdots, n$, the Black-Scholes pricing equation can be written in the form:
\begin{equation}
\frac{\partial V}{\partial t} + \frac{1}{2} \sum_{i = 1}^n
\sum_{j = 1}^n \rho_{ij} \sigma_i \sigma_j S_i S_j
\frac{\partial^2 V}{\partial S_i \partial S_j} + r \sum_{i = 1}^n
S_i \frac{\partial V}{\partial S_i} - rV =0,\quad V(T,S)=\Phi(S).
\label{index-bs}
\end{equation}
Both equations (\ref{eq6}) and (\ref{index-bs}) can be transformed into equation (\ref{PDE-u}) defined on the whole space $\R^n$ (c.f.,  \v{S}ev\v{c}ovi\v{c}, Stehl\'{\i}kov\'a, Mikula \cite[Chapter 4, Section 5]{SSMbook}).

However, stock markets observations indicate that the models (\ref{eq6}) and (\ref{index-bs}) were derived under several restrictive assumptions, e.g perfect replication of a portfolio and its liquidity, completeness and frictionless of the financial market, and/or absence of transaction costs. However, these assumptions are often violated in financial markets. In the past decades, effects of nontrivial transaction costs were investigated by Leland \cite{leland1985option}, Kwok \cite{NBS5}, Avellaneda and Paras \cite{NBS7},  \v{S}ev\v{c}ovi\v{c} and \v{Z}it\v{n}ansk\'a \cite{sevcoviczitnanska}, and many other authors. Feedback and illiquid market effects due to large traders choosing given stock-trading strategies were investigated by Sch\"onbucher and Willmott \cite{NBS13}, Frey and Patie \cite{NBS11}, Frey and Stremme \cite{NBS10}. Effects of the risk arising from an unprotected portfolio were analyzed by Janda\v{c}ka and \v{S}ev\v{c}ovi\v{c} \cite{NBS1}. Option pricing models based on utility maximization were analyzed by Barles and Soner \cite{barles}. A common feature of these generalizations of the linear Black--Scholes equation (\ref{eq6}) is that the constant volatility $\sigma$ is replaced by a nonlinear function depending on the second derivative $\partial _S^2 V$ of the option price $V$. Among nonlinear generalizations of the  Black--Scholes equation, there is the nonlinear Black--Scholes model derived by Frey and Stremme \cite{NBS1} (see also \cite{NBS11, Frey98}). The underlying asset dynamics takes into account the presence of feedback effects due to influence of a large trader choosing particular stock-trading strategy (see also \cite{NBS13}). 

In a recent paper \cite{NBS19}, Cruz and \v{S}ev\v{c}ovi\v{c} generalized the Black--Scholes equation in two important directions. Firstly, following the ideas Frey and Stremme \cite{NBS1}, the model incorporates the effect of a large trader. Secondly, Cruz and \v{S}ev\v{c}ovi\v{c} relaxed the assumption on liquidity of market by assuming that the  underlying asset price follows a L\'evy stochastic process with jumps. The resulting governing equation is the following nonlinear PIDE: 
\begin{equation}
\frac{\partial V}{\partial t}+\frac{1}{2}\frac{\sigma^2 S^2}{\left(1-\varrho S\partial_S \phi\right)^{2}} \frac{\partial^2 V}{\partial S^2 } +r S\frac{\partial V}{\partial S}-rV
+\int_{\mathbb{R}} \left( V(t,S+H)-V(t,S)-H \frac{\partial V}{\partial S}\right) \nu(\ud z)=0,
\label{nonlinearPIDE_one-intro-nonlin}
\end{equation}
where the shift function $H=H(\phi,S,z)$ depends on the large investor stock-trading strategy function $\phi=\phi(t,S)$ as a solution to the following implicit algebraic equation:
\begin{equation}
H =\rho S ( \phi(t, S+H)-\phi(t,S) ) + S(e^{z}-1).
\label{Hfunction}
\end{equation}
The large trader strategy function $\phi$ may depend on the derivative $\partial_S V$ of the option price $V$, e.g., $\phi(t,S)=\partial_S V(t,S)$. Nevertheless, in our application, we assume the trading strategy function $\phi(t,S)$ is prescribed and it is globally H\"older continuous. 

If $\rho=0$, then $H =S(e^{z}-1)$. Using the standard transformation $\tau=T-t,x=\ln(\frac{S}{K})$ and setting $V(t,S)=e^{-r\tau} u(\tau,x)$, then equation \eqref{nonlinearPIDE_one-intro-nonlin} can be reduced to a linear PIDE of the form (\ref{PDE-u}) in the one-dimensional space ($n=1$). 

If $\rho >0$, then \eqref{nonlinearPIDE_one-intro-nonlin} can be transformed into a nonlinear parabolic PIDE. Indeed, suppose that the transformed large trader stock-trading strategy $\psi(\tau,x)=\phi(t,S)$. Then, $V(t,S)$ solves equation  \eqref{nonlinearPIDE_one-intro-nonlin} if and only if the transformed function $u(\tau,x)$ is a solution to the nonlinear parabolic equation:
\begin{eqnarray}
\frac{\partial u}{\partial \tau}&=&\frac{\sigma^2}{2}\frac{1}{(1-\rho \partial_x\psi)^2}
\frac{\partial^2 u}{\partial^2 x}
+\left(r-\frac{\sigma^2}{2}\frac{1}{(1-\rho \partial_x\psi)^2} -  \delta(\tau, x) \right)\frac{\partial u}{\partial x}
\nonumber
\\
&+&\int_{\mathbb{R}} \left( u(\tau,x+\xi)-u(\tau,x)- \xi \frac{\partial u}{\partial x}(\tau,x) \right) \nu(\ud z), 
\quad u(0,x) = \Phi(K e^{x}),
\label{nonlinearPIDEaltsimplified}
\end{eqnarray}
$\tau\in [0,T], x\in \mathbb{R}$. The shift function $\xi(\tau, x, z)$ is a solution to the algebraic equation: 
\begin{equation}
e^\xi  = e^z + \rho (\psi(\tau, x+\xi)-\psi(\tau, x)),
\label{xifunction}
\end{equation}
and $\delta(\tau, x) =\int_{\mathbb{R}} (e^\xi -1 -\xi ) \nu(\ud z) 
=\int_{\mathbb{R}} (e^z -1 - \xi +\rho  (\psi(\tau, x+\xi)-\psi(\tau, x)) ) \nu(\ud z)$. For small values of $0<\rho\ll 1$, we can construct the first order asymptotic expansion $\xi(\tau, x, z) = \xi_0(\tau, x, z) + \rho \xi_1(\tau, x, z)$. For $\rho=0$, we obtain $\xi_0(\tau, x, z)=z$. Hence
\[
e^{z+\rho\xi_1}  = e^z + \rho (\psi(\tau, x+z+ \rho \xi_1)-\psi(\tau, x)).
\]
Taking the first derivative of the above implicit equation with respect to $\rho$ and evaluating it at the origin $\rho=0$, we obtain $\xi_1=e^{-z}(\psi(\tau, x+z)-\psi(\tau, x))$, i.e.,  
\begin{equation}
\xi(\tau, x, z) = z + \rho e^{-z} (\psi(\tau, x+z)-\psi(\tau, x)).
\label{xifunctionexplicit}
\end{equation}

As a consequence, we obtain the following lemma. 

\begin{lemma}\label{holdercontinuityxi}
Assume that the stock-trading strategy $\phi=\phi(t,S)$ is a globally $\omega$-H\"older continuous function, $0<\omega\le 1$. Then, the transformed function $\psi(\tau,x)=\phi(t,S)$ is $\omega$-H\"older continuous, and the first order asymptotic expansion $\xi(\tau, x, z)$ of the nonlinear algebraic equation (\ref{xifunction}) is $\omega$-H\"older continuous in all variables. Furthermore, there exists a constant $C_0>0$ such that  $\sup_{\tau,x}|\xi(\tau,x,z)| \le C_0 |z|^\omega(1+e^{|z|})$ for any $z\in\R$.
\end{lemma}

In what follows, we shall consider a simplified linear approximation of  (\ref{nonlinearPIDE_one-intro-nonlin}) where we set $\rho=0$ in the diffusion function, but we keep the shift function $H$ depending on the parameter $\rho$. Then, the transformed Cauchy problem for the solution $u$  with the first order approximation of the shift function $\xi$ reads as follows:
\begin{eqnarray}
\frac{\partial u}{\partial \tau}&=&\frac{\sigma^2}{2} \frac{\partial^2 u}{\partial^2 x}
+\left(r-\frac{\sigma^2}{2} + \delta(\tau, x) \right)\frac{\partial u}{\partial x}
\nonumber 
\\
&&+\int_{\mathbb{R}} \left( u(\tau,x+\xi)-u(\tau,x)- \xi \frac{\partial u}{\partial x}(\tau,x) \right) \nu(\ud z),
\label{transformedeq}
\end{eqnarray}
$\tau\in [0,T], x\in \mathbb{R}$, where $\xi(\tau, x, z) = z + \rho (\psi(\tau, x+z))-\psi(\tau, x))$. 

Note that the call/put option pay-off functions  $\Phi(S)=\Phi(K e^x)=(S-K)^+ = K (e^x-1)^+$ / $\Phi(S)=\Phi(K e^x)=(K-S)^+=K(1-e^x)^+$ do not belong to the Banach space $X^\gamma$. Following \cite{CruzSevcovic2020}, the procedure on how to overcome this problem and formulate existence and uniqueness of a solution to the PIDE  (\ref{transformedeq}) is based on shifting the solution $u$ by $u^{BS}$. Here, $u^{BS}(\tau,x) =e^{r\tau} V^{BS}(T-\tau, Ke^x)$ is an explicitly given solution to the linear Black-Scholes equation without the PIDE part. That is, $u^{BS}$ solves the linear parabolic equation:
\begin{equation}
\frac{\partial u^{BS}}{\partial \tau} - \frac{\sigma^2}{2} \frac{\partial^2 u^{BS}}{\partial x^2} 
-  \left(r-\frac{\sigma^2}{2}\right)\frac{\partial u^{BS}}{\partial x} =0,
\quad u^{BS}(0,x)=\Phi(K e^{x}), \ \tau\in(0,T), x\in \mathbb{R}.
\label{PDE-uBS}
\end{equation}
Recall that $u^{BS}(\tau,x) = K e^{x+r \tau} N(d_1) - K N(d_2)$ (call option case), where $d_{1,2} = ( x+ (r\pm\sigma^2/2)\tau) /(\sigma\sqrt{\tau})$ (c.f.,  \cite{SSMbook}, \cite{NBS5}). Here,  $N(d)=\frac{1}{\sqrt{2\pi}}\int_{-\infty}^d e^{-\xi^2/2} \ud \xi$ is the CDF of the normal distribution. 

\begin{theorem}\label{existence_linear_PIDE}
Assume the transformed stock-trading strategy function $\psi(\tau,x)$ is globally $\omega$-H\"older continuous in both variables. Suppose that $\nu$ is a L\'evy  measure with the shape parameters $\alpha<3, D\in\R$, where either $\mu>0$, or $\mu=0$ and $D>1$. Let $X^\gamma$ be the space of Bessel potentials space ${\mathscr L}^p_{2\gamma}(\mathbb{R})$, where  $\frac{\alpha-1}{2\omega}<\gamma < \frac{p+1}{2p}$ and  $\frac12 \le \gamma<1$. Let $T>0$. Then, the linear PIDE (\ref{transformedeq}) has a unique mild solution $u$ with the property that the difference $U=u-u^{BS}$ belongs to the space $C([0,T],X^{\gamma})$. 
\end{theorem}

\begin{proof}
First, we outline the idea of the proof. The initial condition $u(0,\cdot)\not\in X^\gamma$ because of two reasons. It is not smooth for $x=0$, and it grows exponentially for $x\to\infty$ (call option) or $x\to-\infty$ (put option). The shifted function $U=u-u^{BS}$ satisfies $U(0,\cdot)\equiv 0$, and so the initial condition $U(0,\cdot)$ belongs to $X^\gamma$. On the other hand, the shift function $u^{BS}$  enters the governing PIDE as it includes the term $f(u^{BS}(\tau,\cdot))$ in the right-hand side. Since $u^{BS}(0,x)$ is not sufficiently smooth for $x=0$, the shift term $f(u^{BS}(\tau,\cdot))$ is singular for $\tau\to 0^+$. Following the ideas of \cite{CruzSevcovic2020}, for the shift term $f(u^{BS}(\tau,\cdot))$, we can provide H\"older estimates which are sufficient for proving the main result of this theorem (c.f.,  \cite[Lemma 4.1]{CruzSevcovic2020}). Furthermore, the exponential growth of the function $u^{BS}$ will be overcome since $\tilde f(e^x) = 0$, where $\tilde f(u) = f(u) - \delta(\tau,\cdot) \partial_x u$, i.e., 
\[
\tilde f(u)(x) =
\int_{\mathbb{R}}\left( u(x+\xi)-u(x)- (e^\xi -1) \partial_x u(x)\, \right) \nu(\ud z).
\]

Next, we present more details of the proof. The function $u^{BS}$ solves the linear PDE (\ref{PDE-uBS}). Thus, the difference $U=u-u^{BS}$ of a solution $u$ to (\ref{transformedeq}) and $u^{BS}$ satisfies the PIDE with the right-hand side:
\begin{eqnarray*}
\frac{\partial U}{\partial \tau}
&=&\frac{\sigma^2}{2} \frac{\partial^2 U}{\partial x^2} 
+ \left(r-\frac{\sigma^2}{2} - \delta(\tau,x)\right)\frac{\partial U}{\partial x} + f(U) + f(u^{BS}) - \delta(\tau,x) \frac{\partial u^{BS}}{\partial x}
\\
&=&\frac{\sigma^2}{2} \frac{\partial^2 U}{\partial x^2} + f(U) + g(\tau, x, \partial_x U) + h(\tau,\cdot),
\end{eqnarray*}
$U(0,x)= 0, \ x\in \mathbb{R}, \tau\in(0,T)$. Here $g(\tau, x, \partial_x U) = (r-\sigma^2/2 - \delta(\tau,x) )\partial_x U$, and $h(\tau,\cdot)=\tilde f(u^{BS}(\tau,\cdot))$. According to Proposition~\ref{prop-f}, $f:X^\gamma\to X$ is a bounded linear mapping. Consequently, it is  Lipschitz continuous, provided that $1/2\le \gamma<1$ and $\gamma>(\alpha-1)/(2\omega)$. Clearly, $\tilde f(e^x)=0$. Hence,
\[
\tilde f(u^{BS}) = \tilde f(u^{BS} - K e^{r\tau+x}), \quad \hbox{and}\ \ 
\partial_\tau \tilde f(u^{BS}) = \tilde f(\partial_\tau(u^{BS} - K e^{r\tau+x})).
\]
Now, it follows from \cite[Lemma 4.1]{CruzSevcovic2020} that the following estimate holds true: 
\[
\Vert  h(\tau_1, \cdot) - h(\tau_2, \cdot) \Vert_{L^p} 
= \Vert \tilde f(u^{BS}(\tau_1, \cdot)) - \tilde f(u^{BS}(\tau_2, \cdot)) \Vert_{L^p} 
\le C_0 |\tau_1-\tau_2|^{-\gamma +\frac{p+1}{2p}}, 
\]
\[
\Vert  h(\tau, \cdot)  \Vert_{L^p} 
= \Vert \tilde f(u^{BS})(\tau, \cdot)) \Vert_{L^p} 
\le C_0 |\tau^{-(2\gamma-1)\left(\frac{1}{2} - \frac{1}{2p}\right)}, 
\]
for any $0<\tau_1,\tau_2,\tau\le T$. The function $h:[0,T]\to X\equiv L^p(\mathbb{R})$ is $((p+1)/(2p)-\gamma)$-H\"older continuous because  $\gamma<\frac{p+1}{2p}$. Moreover, 
\[
\int_0^T \Vert h(\tau, \cdot ) \Vert_{L^p} d\tau=
\int_0^T \Vert \tilde f(u^{BS}(\tau, \cdot)) \Vert_{L^p}d\tau 
\le C_0 \int_0^T\tau^{-(2\gamma-1)\left(\frac{1}{2} - \frac{1}{2p}\right)}d\tau <\infty,
\]
because $(2\gamma-1)\left(\frac{1}{2} - \frac{1}{2p}\right)<1$. We recall that the crucial part of the proof of \cite[Lemma 4.1]{CruzSevcovic2020} was based on the estimates:
\[
\Vert \tilde f(u^{BS}(\tau, \cdot)) \Vert_{L^p} 
\le C_0 \Vert v(\tau,\cdot) \Vert_{X^{\gamma-1/2}},
\quad\text{and}\ 
\Vert \partial_\tau \tilde f(u^{BS}(\tau, \cdot)) \Vert_{L^p} 
\le C_0 \Vert \partial_\tau v(\tau,\cdot) \Vert_{X^{\gamma-1/2}}, 
\]
where $v(\tau,x)= \partial_x \left(u^{BS}(\tau,x) - K e^{r\tau+x}\right)
= K e^{r\tau+x}( N(d_1(\tau,x)) -1)$. This estimate is fulfilled because of Proposition~\ref{prop-f} under the assumptions made on $\gamma$. The proof for the case of a put option is similar. The final estimate on the H\"older continuity of the mapping $h$ follows from careful estimates of the solution $u^{BS}$ derived in the proof of \cite[Lemma 4.1]{CruzSevcovic2020}. The proof now follows from Theorem~\ref{semilinear_existence_result} and Proposition~\ref{semilinear_general_existence_result}. 
\end{proof}

\section{Conclusions}
In this paper, we investigated the existence and uniqueness of a solution to the non-local nonlinear partial integro-differential equation (PIDE) arising from financial modeling. We considered a call/put option pricing model on underlying asset that follows a L\'evy process with jumps in the multidimensional space. We employed the theory of abstract semilinear parabolic equation to obtain the existence and uniqueness of a solution to the PIDE in the scale of Bessel potential Sobolev spaces. We generalized existing results for a general L\'evy measures that satisfies some suitable growth conditions. As an application to option pricing in one-dimensional space, we considered Black-Scholes models for pricing call and put options assuming that the asset follows a L\'evy process. We obtained solutions of the governing nonlinear PIDE where the shift function depends on the large investor stock-trading strategy function, which is a solution to a nonlinear algebraic equation.

\vspace{6pt} 

\noindent{\bf Acknowledgments}
The authors gratefully acknowledge the contribution of the Slovak Research and Development Agency under the project APVV-20-0311 (D.\v{S}.). The research was partially supported by the bilateral German-Slovakian  DAAD Project ENANEFA (C.U.).

\end{document}